\documentclass[pdflatex,sn-mathphys-num]{sn-jnl}% Math and Physical Sciences Numbered Reference Style
\usepackage{soul}
\usepackage{graphicx}%
\usepackage{multirow}%
\usepackage{amsmath,amssymb,amsfonts}%
\usepackage{amsthm}%
\usepackage{mathrsfs}%
\usepackage[title]{appendix}%
\usepackage{xcolor}%
\usepackage{textcomp}%
\usepackage{manyfoot}%
\usepackage{booktabs}%
\usepackage{algorithm}%
\usepackage{algorithmicx}%
\usepackage{algpseudocode}%
\usepackage{listings}%
\usepackage{tikz}
\usetikzlibrary{shapes,arrows,positioning}
\usepackage{tabularx} % para a primeira sugestão
\usepackage{graphicx} % para a segunda sugestão
\usepackage{enumitem}
\usepackage{booktabs}

\usepackage{tabulary}
\theoremstyle{thmstyleone}%
\newtheorem{theorem}{Theorem}%  meant for continuous numbers
\newtheorem{corollary}[theorem]{Corollary}

\theoremstyle{thmstyletwo}%

\theoremstyle{thmstylethree}%

\begin{document}

\title[Article Title]{Training-Free Certified Bounds for Quantum Regression: A Scalable Framework}


% \author*[1,2]{\fnm{Demerson} \sur{N. Gonçalves}}\email{demerson.goncalves@cefet-rj.br}


% \author[1,2]{\fnm{Forth} \sur{Author}}\email{iiiauthor@gmail.com}
\author*[1]{\fnm{Demerson} \sur{N. Gonçalves}}\email{demerson.goncalves@cefet-rj.br}

\author[2]{\fnm{Tharso}  \sur{D. Fernandes}}\email{tharso.fernandes@ufes.br}

\author[3]{\fnm{Pedro  } \sur{H. G. Lugao}}\email{pedro.lugao@cefet-rj.br}
\author[4]{\fnm{João  } \sur{T. Dias}}\email{joao.dias@cefet-rj.br}

% \affil*[1]{\orgdiv{Department of Mathematics}, \orgname{Federal Center for Technological Education of Rio de Janeiro}, \orgaddress{\city{Petrópolis}, \postcode{25620-003}, \state{RJ}, \country{Brazil}}}

% \affil[2]{\orgdiv{Department}, \orgname{Organization}, \orgaddress{\street{Street}, \city{City}, \postcode{10587}, \state{State}, \country{Country}}}

% \affil[3]{\orgdiv{Department}, \orgname{Organization}, \orgaddress{\street{Street}, \city{City}, \postcode{610101}, \state{State}, \country{Country}}}
\affil*[1]{Dept. of Mathematics, Federal Center for Technological Education Celso Suckow da Fonseca (CEFET-RJ), Petrópolis, RJ, Brazil}
\affil[2]{Dept. of Mathematics, Federal University of Espírito Santo (UFES), Alegre, ES, Brazil}
\affil[3]{Dept. of Computer Engineering, CEFET-RJ, Petrópolis, RJ, Brazil}
\affil[4]{Dept. of Telecommunications Engineering, CEFET-RJ, Rio de Janeiro, RJ, Brazil}


\abstract{We present a training-free, certified error bound for quantum regression derived directly from Pauli expectation values. Generalizing the heuristic of minimum accuracy from classification to regression, we evaluate axis-aligned predictors within the Pauli feature space. We formally prove that the optimal axis-aligned predictor constitutes a rigorous upper bound on the minimum training Mean Squared Error (MSE) attainable by any linear or kernel-based regressor defined on the same quantum feature map. Since computing this exact bound requires an intractable scan of the full Pauli basis, we introduce a Monte Carlo framework to efficiently estimate it using a tractable subset of measurement axes. We further provide non-asymptotic statistical guarantees to certify performance within a practical measurement budget. This method enables rapid comparison of quantum feature maps and early diagnosis of expressivity, allowing for the informed selection of architectures before deploying higher-complexity models.}




\maketitle

\section{Introduction}
\label{sec:intro}

The capacity of machine learning (ML) to extract actionable patterns from high-dimensional data has established it as a standard paradigm for decision-making under uncertainty and automated discovery in complex systems~\cite{russel2010, bishop2006pattern}. Quantum machine learning (QML) is a rapidly growing field that aims to transcend classical ML limitations by exploiting the unique structure and computational resources of quantum systems~\cite{biamonte2017quantum, Rocchetto2018}. One of the most promising QML paradigms relies on quantum kernel methods. Central to this framework is the \emph{feature map}, a transformation that embeds classical input data into a high-dimensional Hilbert space to render complex, non-linear relationships amenable to linear analysis. In the quantum setting, this map encodes data into quantum states via parameterized quantum circuits, allowing classical algorithms to operate on the resulting quantum-induced feature space. This hybrid quantum--classical framework was established in the foundational works of Havlíček et al.~\cite{havlicek2019supervised}, Schuld and Killoran~\cite{schuld2019quantum}, and Mitarai et al.~\cite{mitarai2018quantum}. Building on these seminal contributions, subsequent research has extensively investigated the hardware implementation and theoretical properties of these maps~\cite{huang2021power, li2022quantumkernel, tacchino2020quantum, gentinetta2024complexity}.

However, leveraging these quantum feature spaces in practice involves significant hurdles. 
In variational approaches, the optimization of parameterized circuits is frequently obstructed by the phenomenon of \emph{barren plateaus}, where the gradients of the cost function vanish exponentially with the number of qubits, rendering deep architecture training effectively impossible~\cite{mcclean2018barren, cerezo2021cost, larocca2024barren}. 
To navigate this landscape without full training, geometric descriptors such as \emph{expressibility} and \emph{entangling capability}~\cite{sim2019expressibility} have been proposed to characterize the ansatz structure. While these task-agnostic metrics provide valuable insights into the statistical distribution of states, they do not directly predict the generalization performance or the training error for a specific target function.
Furthermore, although quantum kernel methods circumvent gradient-based optimization, they face severe scalability constraints in the Noisy Intermediate-Scale Quantum (NISQ) era~\cite{preskill2018quantum}, where limited coherence times and hardware noise make the estimation of full kernel matrices $N \times N$ prohibitively expensive. Validating a quantum model also requires benchmarking against robust classical baselines, such as standard Support Vector Regression (SVR)~\cite{smola2004tutorial}, which sets a high threshold for performance.

Consequently, assessing and comparing quantum feature maps by fully training models for every candidate architecture becomes computationally intractable. 
In kernel-based regression, this typically necessitates constructing the full kernel matrix and tuning regularization hyperparameters; in variational approaches, it requires iterative circuit evaluations for gradient updates. These costs scale poorly with dataset size and circuit depth. 
To address similar challenges in classification, Suzuki et al.~\cite{suzuki2022analysis} proposed the \emph{minimum accuracy} heuristic. 
This method estimates the baseline performance obtainable by restricting measurements to axis-aligned Pauli observables, avoiding explicit optimization. 
However, the original formulation focuses exclusively on classification and relies on a full scan of the Pauli basis. 
This creates two significant gaps: first, there is no direct analogue for regression tasks, where minimizing continuous error, such as the mean squared error (MSE), is paramount; second, the computational cost of a full Pauli decomposition scales as $4^n$ for $n$ qubits, rendering the exact metric intractable for larger systems.

In this work, we extend the analysis and practical utility of feature map diagnostics in three main directions, generalizing the axis-aligned concept to quantum regression. 
First, we formally define a regression setting in the Pauli feature space and introduce a training-free score based on single-coordinate least squares. 
Unlike geometric heuristics, this metric is directly grounded in the minimization of the MSE. 
Second, we provide a rigorous theoretical result (Theorem~\ref{thm:axis-upper-bound}) proving that the MSE of the best axis-aligned predictor constitutes a certified \emph{upper bound} on the minimum training MSE attainable by any linear or kernel-based regressor on the same feature map. 
This formally justifies the score as a ``guaranteed limit'': if the axis-aligned error is low, the full quantum model is mathematically guaranteed to perform at least as well. 
Third, we introduce a statistically-certified Monte Carlo (MC) framework to render this bound computationally viable. 
Instead of the intractable $4^n$ scan, we introduce an estimator, $\widehat{\mathrm{MSE}}_{\mathrm{axis}}$, computed from a random subset of axes. 
We derive non-asymptotic probability guarantees based on Hoeffding's concentration inequalities~\cite{hoeffding1963probability}, relating the sample size to the reliability of the bound (Theorem~\ref{thm:coverage}).

From a practical standpoint, our results transform this theoretical bound into a scalable tool for pre-screening quantum feature maps. 
The proposed workflow allows one to avoid the exponential cost of characterizing the full feature space. 
By drawing a random sample of $t \ll 4^n$ Pauli axes and computing simple 1D regression scores, one obtains a certified upper bound on the model's potential error. 
This enables the rapid comparison of architectures, early diagnosis of expressivity issues, and informed resource allocation, effectively filtering out poor feature maps before deploying higher-complexity models. 
Consequently, the computational burden of feature map selection is drastically reduced from the polynomial complexity of full kernel training (typically $\mathcal{O}(N^3)$) to a scalable cost dominated by the number of sampled axes $t$, where $t$ can be kept small while maintaining rigorous statistical confidence.

% The remainder of this paper is organized as follows. 
% Section~\ref{sec:bound} introduces the Pauli-axis upper bound for quantum regression, formalizing the axis-aligned least-squares problem and establishing the certified upper bound.
% Section~\ref{sec:analytical-derivation} provides an analytical derivation and discusses the computational complexity, motivating the need for approximation.
% Section~\ref{sec:monte-carlo} addresses this challenge by introducing the Monte Carlo estimation framework and deriving non-asymptotic statistical guarantees (Section~\ref{sec:Statisticalguarantees}).
% Section~\ref{sec:adaptive-p} presents the adaptive sample-size calibration algorithm.
% Finally, Section~\ref{sec:numerical} validates the proposed framework through numerical experiments on synthetic and benchmark datasets, demonstrating both the effectiveness and the practical scalability of the method.
The remainder of this paper is organized as follows.
Section~\ref{sec:bound} establishes the theoretical framework, formally defining the axis-aligned regression score and proving it constitutes a certified upper bound on the optimal training error.
Section~\ref{sec:monte-carlo} introduces the Monte Carlo estimation strategy designed to overcome the exponential scaling of the feature space; this section also derives rigorous non-asymptotic statistical guarantees and details the adaptive algorithm for sample-size calibration.
The experimental methodology, including dataset generation and feature map configurations, is detailed in Section~\ref{sec:setup}.
Section~\ref{sec:numerical} presents the numerical benchmarks, validating the estimator's tightness and predictive power against fully trained quantum and classical regressors.
Finally, Section~\ref{sec:conclusions} concludes the work and outlines directions for future research.

\section{Pauli-axis Upper Bound for Quantum Regression}
\label{sec:bound}

We consider supervised regression in the quantum kernel framework, where the goal is to predict a real-valued label \( y \in \mathbb{R} \) from an input \( \mathbf{x} \in \mathcal{X} \subseteq \mathbb{R}^m \), given a dataset \( D = \{(\mathbf{x}_k, y_k)\}_{k=1}^N \). 
In this approach, classical data are embedded into an \( n \)-qubit Hilbert space through a quantum feature map \( |\Phi(\mathbf{x})\rangle = U(\mathbf{x})\,|0\rangle^{\otimes n} \), where \( U(\mathbf{x}) \) is a parameterized encoding circuit. 
The similarity between two encoded inputs is quantified by the fidelity kernel
\begin{equation}
K(\mathbf{x},\mathbf{z}) = \big|\langle \Phi(\mathbf{x}) \mid \Phi(\mathbf{z}) \rangle\big|^2,
\end{equation}
which can be estimated on a quantum device via the transition probability of the corresponding circuit~\cite{havlicek2019supervised}, while the regression model itself is optimized classically.

To analyze the geometry of this mapping, we consider the corresponding density operators \( \rho(\mathbf{x}) = |\Phi(\mathbf{x})\rangle\langle\Phi(\mathbf{x})| \). 
While the underlying \( n \)-qubit Hilbert space has complex dimension \( 2^n \), the space of Hermitian operators acting on it constitutes a real vector space. 
This operator space is naturally spanned by the set of \( n \)-qubit Pauli operators \(\mathcal{P}_n = \{I,X,Y,Z\}^{\otimes n}\)~\cite{nielsen2010quantum}. 
Consequently, any pure-state density matrix can be expanded in this basis as \( \rho(\mathbf{x}) = \frac{1}{2^n}\sum_{i=1}^{d} a_i(\mathbf{x})\,\sigma^i \), where \( \sigma^i \in \mathcal{P}_n \) and \( a_i(\mathbf{x}) = \mathrm{tr}[\rho(\mathbf{x})\,\sigma^i] \) is the expectation value of the \(i\)-th Pauli operator. 
Using the orthogonality relation \( \mathrm{tr}(\sigma^i\sigma^j) = 2^n\delta_{ij} \), the quantum kernel can be rewritten as
\begin{equation}
K(\mathbf{x},\mathbf{z}) = \frac{1}{2^n}\sum_{i=1}^{d} a_i(\mathbf{x})\,a_i(\mathbf{z}),
\end{equation}
where \(d\) denotes the total dimension of the Pauli feature space, typically \(d = 4^n\) when the full Pauli basis is considered.
This reformulation reveals that the quantum kernel is equivalent to a standard linear kernel in a real feature space of dimension \( 4^n \), providing a direct geometric bridge between quantum and classical descriptions.

\subsection{Axis-Aligned Least-Squares Upper Bound}
\label{sec:axis-bound}

Building upon this representation, we now formalize the least-squares regression model. 
Each encoded input \(\mathbf{x}_k\) is mapped to a feature vector \(\mathbf{a}(\mathbf{x}_k) = [a_1(\mathbf{x}_k), \dots, a_d(\mathbf{x}_k)]\). 
Given a set of target labels \( \{y_k\}_{k=1}^N \), the quality of any predictor \( f: \mathcal{X} \to \mathbb{R} \) is measured by the empirical MSE, following the standard framework of empirical risk minimization~\cite{bishop2006pattern}:
\begin{equation}
\mathrm{MSE}(f) = \frac{1}{N}\sum_{k=1}^N \big(y_k - f(\mathbf{x}_k)\big)^2.
\label{eq:mse-emp}
\end{equation}
We define the family of full affine linear regressors in the Pauli feature space as \( \mathcal{F} = \{ f(\mathbf{x}) = \langle \mathbf{w}, \mathbf{a}(\mathbf{x}) \rangle + b \mid \mathbf{w} \in \mathbb{R}^{d},\; b \in \mathbb{R} \} \). 
Ideally, we seek the predictor in \(\mathcal{F}\) that minimizes Eq.~\eqref{eq:mse-emp}, achieving the optimal training error \(\mathrm{MSE}^{*} = \min_{f\in\mathcal{F}} \mathrm{MSE}(f)\). 
However, characterizing this minimum generally requires full access to the exponentially large feature space. 
As a tractable alternative, we consider the subclass of \emph{axis-restricted} predictors, which depend on a single Pauli feature \(i \in \{1, \dots, d\}\):
\begin{equation}
\mathcal{F}^{(i)}_{\mathrm{axis}} = \Big\{\, f_i(\mathbf{x}) = w\,a_i(\mathbf{x}) + b \;\Big|\; w,b \in \mathbb{R} \,\Big\}.
\end{equation}
We denote the minimum error achievable by exploiting only the \(i\)-th axis as \(\mathrm{MSE}_i = \min_{f\in\mathcal{F}^{(i)}_{\mathrm{axis}}} \mathrm{MSE}(f)\), and the best overall single-axis performance as \(\mathrm{MSE}_{\mathrm{axis}} = \min_{i} \mathrm{MSE}_i\).

The nested structure of these hypothesis classes leads directly to a certified performance guarantee.

\begin{theorem}[Pauli-axis upper bound]
\label{thm:axis-upper-bound}
For any dataset \(D\) encoded by a fixed quantum feature map, the optimal affine linear regression error is upper bounded by the best axis-aligned error:
\[
\mathrm{MSE}^{*} \le \mathrm{MSE}_{\mathrm{axis}}.
\]
\end{theorem}

\begin{proof}
Since \(\mathcal{F}^{(i)}_{\mathrm{axis}} \subset \mathcal{F}\) for every coordinate \(i\), the minimum over the larger set \(\mathcal{F}\) cannot exceed the minimum over the subset \(\mathcal{F}^{(i)}_{\mathrm{axis}}\). 
Therefore, \(\mathrm{MSE}^{*} \le \mathrm{MSE}_i\) for all \(i\). 
Taking the minimum over all \(d\) coordinates yields \(\mathrm{MSE}^{*} \le \min_i \mathrm{MSE}_i = \mathrm{MSE}_{\mathrm{axis}}\).
\end{proof}

This result extends naturally to more complex hypothesis classes \(\mathcal{H} \supseteq \mathcal{F}\), such as those used in kernel ridge regression or support vector regression (SVR) on the induced feature space. 
Since these methods optimize over a space that includes all linear functions, their empirical error satisfies \(\inf_{h\in\mathcal{H}}\mathrm{MSE}(h) \le \mathrm{MSE}^{*}\). 
Consequently, \(\mathrm{MSE}_{\mathrm{axis}}\) serves as a conservative, training-free proxy: if a simple axis-aligned model achieves low error, the more expressive kernel model is mathematically guaranteed to perform at least as well.
However, to utilize this bound as a diagnostic tool, we require an explicit, computable expression for \(\mathrm{MSE}_{\mathrm{axis}}\) in terms of observable data statistics.

\subsection{Analytical Derivation and Complexity}
\label{sec:analytical-derivation}

We now derive an explicit analytical expression for the bound in Theorem~\ref{thm:axis-upper-bound}. 
For any fixed Pauli coordinate \(i\), the optimal predictor \(f \in \mathcal{F}^{(i)}_{\mathrm{axis}}\) corresponds to the univariate ordinary least squares (OLS) solution. 
Its parameters depend solely on the empirical moments of the data. 
Let us define the sample means, variances, and covariances as:
\begin{align}
\bar{y} &= \frac{1}{N}\sum_{k=1}^N y_k, & \bar{a}_i &= \frac{1}{N}\sum_{k=1}^N a_i(\mathbf{x}_k), \nonumber \\
\mathrm{Var}(a_i) &= \frac{1}{N}\sum_{k=1}^N \big(a_i(\mathbf{x}_k)-\bar{a}_i\big)^2, & \mathrm{Cov}(a_i,y) &= \frac{1}{N}\sum_{k=1}^N \big(a_i(\mathbf{x}_k)-\bar{a}_i\big)\big(y_k-\bar{y}\big).
\end{align}
If \(\mathrm{Var}(a_i)>0\), the optimal coefficients \(w_i^{\star}\) and \(b_i^{\star}\) are given by the standard OLS result:
\begin{equation}
w_i^{\star}=\frac{\mathrm{Cov}(a_i,y)}{\mathrm{Var}(a_i)},\qquad
b_i^{\star}=\bar{y}-w_i^{\star}\bar{a}_i.
\end{equation}
Substituting these into the error function yields the residual variance:
\begin{equation}
\mathrm{MSE}_i = \mathrm{Var}(y) - \frac{\mathrm{Cov}(a_i,y)^2}{\mathrm{Var}(a_i)}
= \mathrm{Var}(y)\big(1-\rho_{y,a_i}^{\,2}\big),
\end{equation}
where \(\rho_{y,a_i}\) is the Pearson correlation coefficient between the target \(y\) and the feature \(a_i\). 
This equality follows from the fundamental property of simple linear regression, where the coefficient of determination \(R^2\) is exactly the square of the correlation coefficient~\cite{kutner2005applied}.
In the degenerate case where \(\mathrm{Var}(a_i)=0\), the predictor reduces to the constant mean \(f_i(\mathbf{x})=\bar{y}\), yielding \(\mathrm{MSE}_i=\mathrm{Var}(y)\).

By substituting this result back into Theorem~\ref{thm:axis-upper-bound}, we obtain a closed-form expression for the global upper bound:
\begin{equation} \label{eq:closed_form}
\mathrm{MSE}_{\mathrm{axis}}
= \min_{1\le i\le d} \mathrm{MSE}_i
= \mathrm{Var}(y)\left(1 - \max_{1\le i\le d}\rho_{y,a_i}^{\,2}\right).
\end{equation}
Eq.~\eqref{eq:closed_form} offers a powerful geometric interpretation: the certified upper bound is determined solely by the single Pauli axis that exhibits the strongest marginal correlation with the target variable.

Regarding computational complexity, evaluating Eq.~\eqref{eq:closed_form} requires iterating over all feature dimensions. 
Ideally, this consumes \(O(Nd)\) time and \(O(d)\) memory. 
However, two practical considerations arise. 
First, when Pauli expectations are estimated from finite quantum shots, the empirical moments inherit sampling noise. 
In statistical terms, this introduces measurement error in the regressors, leading to \emph{attenuation bias}~\cite{fuller1987measurement}: the noise artificially inflates \(\mathrm{Var}(a_i)\) and systematically shrinks \(|\rho_{y,a_i}|\). 
As a result, the computed \(\mathrm{MSE}_{\mathrm{axis}}\) becomes conservative (larger than the ideal noiseless value), but preserves the validity of the upper bound.

Second, and more critically, the dimension \(d\) corresponds to the full operator basis size, \(d=4^n\). 
While the calculation is linear in \(d\), the exponential growth of the basis with the number of qubits renders the exact evaluation of \(\max_i \rho_{y,a_i}^2\) intractable for multi-qubit systems. 
This scaling bottleneck motivates the introduction of the Monte Carlo estimation framework presented in the next section.

\section{Monte Carlo Estimation Framework}
\label{sec:monte-carlo}

To overcome the exponential scaling of the Pauli basis ($d=4^n$) identified in Section~\ref{sec:bound}, we introduce a probabilistic estimation method. 
A similar Monte Carlo strategy has been recently proposed to estimate the minimum accuracy of quantum classifiers~\cite{goncalves2025certified}. 
In this work, we extend and formalize this framework for the continuous regression domain, deriving specific concentration bounds for the MSE.

Instead of scanning all $d$ coordinates, we uniformly sample a random subset of axes \( T \subset \{1, \dots, d\} \) of cardinality \( t \), where \( t \ll d \). 
We define the Monte Carlo estimator as the minimum error found within that subset:
\begin{equation}
\widehat{\mathrm{MSE}}_{\mathrm{axis}}(T) := \min_{i \in T} \mathrm{MSE}_i = \mathrm{Var}(y)\left(1 - \max_{i \in T}\rho_{y,a_i}^{\,2}\right).
\label{eq:mc-estimator}
\end{equation}
This estimator selects the best axis from the subset \(T\), providing an efficiently computable surrogate for \(\mathrm{MSE}_{\mathrm{axis}}\). 
Crucially, despite being an approximation, this estimator preserves the rigorous performance guarantee of the exact metric, as established by the following theorem.

\begin{theorem}[Monte Carlo bound]
\label{teo:MCbound}
For any non-empty subset \( T \subseteq \{1, \dots, d\} \), the optimal affine linear error \(\mathrm{MSE}^*\) is upper bounded by the Monte Carlo estimator:
\begin{equation}
\mathrm{MSE}^* \le \mathrm{MSE}_{\mathrm{axis}} \le \widehat{\mathrm{MSE}}_{\mathrm{axis}}(T).
\label{eq:mc-bound}
\end{equation}
\end{theorem}

\begin{proof}
By definition, \(\mathrm{MSE}_{\mathrm{axis}} = \min_{1 \le j \le d} \mathrm{MSE}_j\). 
Since the minimum over a subset \(T\) is necessarily greater than or equal to the minimum over the full set (i.e., \(T \subseteq \{1, \dots, d\} \implies \min_{i \in T} x_i \ge \min_{all} x_i\)), it follows that \(\mathrm{MSE}_{\mathrm{axis}} \le \widehat{\mathrm{MSE}}_{\mathrm{axis}}(T)\). 
Combining this with Theorem~\ref{thm:axis-upper-bound}, which states \(\mathrm{MSE}^* \le \mathrm{MSE}_{\mathrm{axis}}\), yields the chain of inequalities \(\mathrm{MSE}^* \le \mathrm{MSE}_{\mathrm{axis}} \le \widehat{\mathrm{MSE}}_{\mathrm{axis}}(T)\).
\end{proof}

The inequalities in (\ref{eq:mc-bound}) confirm that the MC estimator serves as a certified upper bound on the optimal error \(\mathrm{MSE}^*\), requiring no model training. 
Furthermore, this bound is systematically tightenable: the map \(T \mapsto \widehat{\mathrm{MSE}}_{\mathrm{axis}}(T)\) is non-increasing with respect to set inclusion. 
Specifically, for nested subsets \(T_t \subset T_{t+1}\), the estimator satisfies \(\widehat{\mathrm{MSE}}_{\mathrm{axis}}(T_{t+1}) \le \widehat{\mathrm{MSE}}_{\mathrm{axis}}(T_t)\), converging to the exact value \(\mathrm{MSE}_{\mathrm{axis}}\) as \(t \to d\).

\subsection{Statistical Guarantees and Threshold Selection}
\label{sec:Statisticalguarantees}

While Theorem~\ref{teo:MCbound} establishes that the Monte Carlo estimator
$\widehat{\mathrm{MSE}}_{\mathrm{axis}}(T)$ is a valid upper bound on the optimal affine regression error,
it does not quantify the probability that this bound is sufficiently tight in practice.
In regression tasks, tightness must be interpreted relative to the intrinsic scale of the target variable.
For this reason, we introduce a target threshold $\tau$ defined in terms of a desired coefficient of determination
$R^2_{\mathrm{target}} \in (0,1]$:
\begin{equation}
\tau = \mathrm{Var}(y)\bigl(1 - R^2_{\mathrm{target}}\bigr).
\label{eq:tau_definition}
\end{equation}
Equivalently, we may write $\tau = \tau_{\mathrm{ratio}} \cdot \mathrm{Var}(y)$, where
$\tau_{\mathrm{ratio}} := 1 - R^2_{\mathrm{target}}$.

Any axis $i$ satisfying $\mathrm{MSE}_i \le \tau$ therefore explains at least a fraction
$R^2_{\mathrm{target}}$ of the sample variance.
To characterize how frequently such axes occur in the Pauli feature space, we define the set of
\emph{good axes}
\begin{equation}
\mathcal{A}_\tau = \{ i \in \{1,\ldots,d\} : \mathrm{MSE}_i \le \tau \},
\label{eq:good-axes}
\end{equation}
and introduce the empirical cumulative distribution function over the finite axis set,
\begin{equation}
F(\tau) = \frac{1}{d}\,|\mathcal{A}_\tau|.
\end{equation}
We denote by $p := F(\tau)$ the fraction of axes achieving the target performance level.

The following result quantifies the probability that a random subset of axes captures at least one good axis.

\begin{theorem}
\label{thm:coverage}
Fix a threshold $\tau$ and let $p = F(\tau)$.
If $T \subset \{1,\ldots,d\}$ is sampled uniformly without replacement with $|T| = t$, then
\begin{equation}
\mathbb{P}\!\left(\widehat{\mathrm{MSE}}_{\mathrm{axis}}(T) \le \tau\right)
= 1 - \frac{\binom{d - \lfloor p d \rfloor}{t}}{\binom{d}{t}}
\;\ge\; 1 - (1-p)^t.
\label{eq:quantile-coverage}
\end{equation}
\end{theorem}

\begin{proof}
Let $\mathcal{A}_\tau$ denote the set of good axes with cardinality
$m = |\mathcal{A}_\tau| = \lfloor p d \rfloor$.
The event $\widehat{\mathrm{MSE}}_{\mathrm{axis}}(T) \le \tau$ occurs if and only if
$T \cap \mathcal{A}_\tau \neq \emptyset$.
Under uniform sampling without replacement, the probability that $T$ avoids $\mathcal{A}_\tau$
is given by the hypergeometric term $\binom{d-m}{t} / \binom{d}{t}$.
The inequality follows from the bound
$\binom{d-m}{t} / \binom{d}{t} \le (1 - m/d)^t = (1-p)^t$.
\end{proof}

An immediate consequence of Theorem~\ref{thm:coverage} is a sufficient condition on the sample size
required to achieve a desired confidence level.

\begin{corollary}
\label{cor:ssize}
Fix $\delta \in (0,1)$ and suppose that at least a fraction $p \in (0,1)$ of the axes satisfies
$\mathrm{MSE}_i \le \tau$.
If $T$ is sampled uniformly without replacement with $|T| = t$, then the condition
\begin{equation}
t \;\ge\; \frac{\log(1/\delta)}{\log\!\bigl(1/(1-p)\bigr)}
\label{eq:ssize-exact}
\end{equation}
ensures
\[
\mathbb{P}\!\left(\widehat{\mathrm{MSE}}_{\mathrm{axis}}(T) \le \tau\right) \ge 1 - \delta.
\]
In particular, since $-\log(1-p) \ge p$ for all $p \in (0,1)$, the simpler sufficient condition
\begin{equation}
t \;\ge\; \frac{1}{p}\,\log\!\frac{1}{\delta}
\label{eq:ssize-simple}
\end{equation}
also guarantees the same confidence level.
\end{corollary}

\begin{proof}
By Theorem~\ref{thm:coverage},
$\mathbb{P}(\widehat{\mathrm{MSE}}_{\mathrm{axis}}(T) \le \tau) \ge 1 - (1-p)^t$.
Requiring the right-hand side to be at least $1-\delta$ yields $(1-p)^t \le \delta$,
which is equivalent to~\eqref{eq:ssize-exact}.
The bound~\eqref{eq:ssize-simple} follows from $-\log(1-p) \ge p$.
\end{proof}

\subsection{Adaptive Sample-Size Calibration}
\label{sec:adaptive-p}

The sample-size bounds in Corollary~\ref{cor:ssize} depend on the unknown fraction
$p = F(\tau)$ of axes satisfying the threshold condition.
Since computing $p$ via an exhaustive scan of the $d = 4^n$ Pauli axes is intractable at scale,
we employ an adaptive Monte Carlo procedure that estimates this quantity online while preserving
rigorous probabilistic guarantees.

Fixing a threshold $\tau$, each sampled axis $i$ is associated with the indicator variable
\begin{equation}
X_i := \mathbf{1}\{\mathrm{MSE}_i \le \tau\},
\label{eq:bernoulli}
\end{equation}
which records whether the axis meets the target performance.
Although sampling is performed without replacement from a finite population,
we employ one-sided Hoeffding bounds, which remain conservative under this regime.

After sampling a subset $T$ of $t$ distinct axes, define
\begin{equation}
s := \sum_{i \in T} X_i,
\qquad
\hat{p} := \frac{s}{t},
\label{eq:empirical-p}
\end{equation}
as the number of successful axes and the corresponding empirical success rate.
For a confidence parameter $\alpha \in (0,1)$, a lower confidence bound on $p$ is given by
\begin{equation}
p_L
:= \max\!\left\{0,\;
\hat{p} - \sqrt{\frac{1}{2t}\log\frac{1}{\alpha}}
\right\},
\label{eq:plower}
\end{equation}
which holds with probability at least $1-\alpha$.

Substituting $p_L$ into the exact coverage condition of Theorem~\ref{thm:coverage}
yields a data-driven estimate of the required sample size,
\begin{equation}
t_{\mathrm{req}}
:= \frac{\log(1/\delta)}{\log\!\bigl(1/(1-p_L)\bigr)},
\label{eq:treq}
\end{equation}
with the convention $t_{\mathrm{req}} = \infty$ when $p_L = 0$.
Sampling proceeds iteratively until the stopping condition $t \ge t_{\mathrm{req}}$ is met,
at which point the estimator satisfies
$\mathbb{P}(\widehat{\mathrm{MSE}}_{\mathrm{axis}}(T) \le \tau) \ge 1 - \delta$.

To prevent unbounded execution in regimes where high-quality axes are extremely sparse,
we impose a maximum sampling budget $t_{\max}$.
In addition, a futility stopping criterion is employed:
if no satisfactory axes are observed ($\hat{p}=0$) and the confidence width falls below a minimal tolerance,
the procedure terminates early.
In such cases, the returned value $\widehat{\mathrm{MSE}}_{\mathrm{axis}}(T)$
remains a valid certified upper bound on the optimal regression error by
Theorem~\ref{teo:MCbound}, albeit without the threshold guarantee.
This adaptive strategy, summarized in Algorithm~\ref{alg:adaptive},
translates the theoretical coverage bounds into a practical and robust stopping rule,
allocating computational effort in proportion to the empirical density of high-performing axes.

\begin{algorithm}[th]
\caption{Adaptive Monte Carlo calibration of the axis-aligned bound}
\label{alg:adaptive}
\begin{algorithmic}[1]
\Require $\tau_{\mathrm{ratio}}$, $\delta_{\mathrm{total}}$, $t_0$, $b$, $t_{\max}$, $\epsilon_{\min}$
\Ensure Sampled set $T$, estimator $\widehat{\mathrm{MSE}}_{\mathrm{axis}}(T)$

\State $\tau \leftarrow \tau_{\mathrm{ratio}} \mathrm{Var}(y)$, \quad $\alpha = \delta = \delta_{\mathrm{total}}/2$
\State Sample initial set $T$ ($|T|=t_0$); \quad $t \leftarrow t_0$, $s \leftarrow \sum_{i\in T} \mathbf{1}\{\mathrm{MSE}_i \le \tau\}$

\While{$t < t_{\max}$}
    \State $\hat p \leftarrow s/t, \quad \epsilon \leftarrow \sqrt{\frac{1}{2t}\log\frac{1}{\alpha}}, \quad p_L \leftarrow \max\{0,\hat p - \epsilon\}$
    \State $t_{\mathrm{req}} \leftarrow {\log(1/\delta)}/{\log(1/(1-p_L))}$
    
    \If{$t \ge t_{\mathrm{req}}$ \textbf{or} ($\hat p = 0 \land \epsilon < \epsilon_{\min}$)} \textbf{break} \EndIf

    \State Update $T$ with $b$ new axes; \quad $s \leftarrow s + \Delta s, \quad t \leftarrow t + b$
\EndWhile

\State \Return $T,\; \min_{i \in T} \mathrm{MSE}_i$
\end{algorithmic}
\end{algorithm}

Algorithm~\ref{alg:adaptive} formalizes the adaptive Monte Carlo procedure used throughout the experimental section, providing a certified stopping rule that avoids exhaustive scans of the $4^n$ axis-aligned hypothesis space.

%%%%%%%%%%%%%%%%%%%%%%%%%%%%%%%%%%%%%%%%%%%%%%%%%%%%%%
\section{Experimental Setup}
\label{sec:setup}

To validate the theoretical framework and the adaptive Monte Carlo procedure, we designed a comprehensive experimental suite varying structural complexity, noise levels, and feature map architectures.

\subsection{Datasets and Benchmarking Strategy}
\label{subsec:datasets}

To rigorously validate the behavior of the adaptive estimator, we employ a hybrid benchmarking strategy. 
First, we construct two synthetic landscapes with diametrically opposed geometric properties, one dense and correlated, the other sparse and high-dimensional, to serve as controlled stress tests. 
Since the ground truth structure of these tasks is known by design, they allow us to verify if the estimator's convergence behavior aligns with theoretical expectations. 
Finally, we challenge the method on a real-world dataset to assess its robustness under unstructured noise.

The three representative regression tasks are summarized in Table~\ref{tab:datasets}:
(i) \texttt{synthetic\_corr\_gauss}, representing a moderate regime with correlated features;
(ii) \texttt{synthetic\_sparse}, a hard regime with high-dimensional sparse dependencies; and
(iii) \texttt{real\_california}, a subset of the California Housing dataset reduced to \(n=4\) principal components, serving as a real-world benchmark with natural noise.

\begin{table}[htbp]
\centering
\caption{Profile of datasets used in the numerical experiments.}
\label{tab:datasets}
\renewcommand{\arraystretch}{1.2}
\begin{tabularx}{\linewidth}{l X}
\toprule
\textbf{Property} & \textbf{Details} \\
\midrule
\textbf{Dataset} & \texttt{synthetic\_corr\_gauss} \\
\textbf{Regime} & \textbf{Moderate (Dense):} Information is spread across correlated axes; high density of near-optimal solutions. \\
\textbf{Target Function} & $y = \exp(-\frac{1}{2}\mathbf{x}^\mathsf{T}\Sigma^{-1}\mathbf{x})$, with $\Sigma_{ij}=0.3+0.7\delta_{ij}$. \\
\midrule
\textbf{Dataset} & \texttt{synthetic\_sparse} \\
\textbf{Regime} & \textbf{Hard (Sparse):} Information is concentrated in a hidden subset; low density of informative axes (needle-in-haystack). \\
\textbf{Target Function} & $y = \sum_{i\in \mathcal{S}}\sin(x_i) + 0.1\sum_{j\notin\mathcal{S}}x_j$, where $\mathcal{S}$ is the active subset. \\
\midrule
\textbf{Dataset} & \texttt{real\_california} \\
\textbf{Regime} & \textbf{Real-world:} Natural data distribution with unstructured noise and complex dependencies. \\
\textbf{Description} & California Housing (PCA reduced to $n=4$, retaining 95\% variance). \\
\bottomrule
\end{tabularx}
\end{table}

\subsection{Quantum Feature Map and Configuration Space}
\label{sec:feature_map_def}

The quantum feature space is generated by the Pauli Feature Map ansatz~\cite{havlicek2019supervised}, which maps classical data \(\mathbf{x}\in\mathbb{R}^n\) into an \(n\)-qubit state via interleaved Hadamard layers and data-dependent entangling gates:
\begin{equation}\label{unitary}
\mathcal{U}_{\Phi(\mathbf{x})}=\big(U_{\Phi(\mathbf{x})}\,H^{\otimes n}\big)^r,
\qquad
U_{\Phi(\mathbf{x})}=\exp\!\left(i\sum_{S\in\mathcal{I}}\phi_S(\mathbf{x})\prod_{k\in S}P_k\right).
\end{equation}
Here, \(P_k\in\{X,Y,Z\}\) are Pauli operators, \(\mathcal{I}\) defines the connectivity graph, and \(\phi_S(\mathbf{x})\) is a classical encoding function.

To ensure robustness and architectural diversity, our study does not rely on a single model. Instead, we conduct a comprehensive evaluation across a total of 324 experimental configurations, comprising 108 distinct feature maps applied to each dataset. We systematically explore this design space by permuting the three key components defined in Table~\ref{tab:featuremap-config}: (1) \emph{classical preprocessing}, to handle data scaling; (2) \emph{Pauli sequences}, determining the basis of quantum correlations; and (3) \emph{data-mapping rules}, controlling the nonlinearity of the phase encoding. These core components are further combined with varying entanglement patterns and circuit depths ($r \in \{1,2\}$) to cover a broad spectrum of model complexity.

\begin{table}[h]
\centering
\caption{Configuration space definition. The Pauli sequences include \texttt{all\_pairs} which comprises all two-body combinations (\texttt{XY}, \texttt{YZ}, \texttt{ZZ}, etc.).}
\label{tab:featuremap-config}
\renewcommand{\arraystretch}{1.3}
\begin{tabularx}{\linewidth}{l >{\raggedright}X >{\raggedright\arraybackslash}X}
\toprule
\textbf{Stage} & \textbf{Option and definition} & \textbf{Qualitative effect} \\
\midrule
\multirow{3}{*}{\shortstack[l]{\textbf{Classical}\\\textbf{preprocessing}}}
  & \texttt{id}: \(f(x)=x\) & Unbounded inputs (raw data) \\
  & \texttt{tanh}: \(f(x)=\tanh(x)\) & Bounded to \([-1, 1]\), outlier-robust \\
  & \texttt{rbf-s1}: \(f(x)=e^{-x^2/2}\) & Local similarity emphasis (Gaussian) \\
\midrule
\multirow{3}{*}{\shortstack[l]{\textbf{Pauli}\\\textbf{sequence}}}
  & \texttt{Z+ZZ} ([Z, ZZ]) & Computational basis correlations \\
  & \texttt{Y+YY} ([Y, YY]) & Off-diagonal interference (complex) \\
  & \texttt{all\_pairs} ([XY, YZ, ZZ...]) & Richer entanglement structure \\
\midrule
\multirow{3}{*}{\shortstack[l]{\textbf{Data-mapping}\\(\(\phi(x)\))}}
  & \texttt{prod}: \(\prod_i x_i\) & Purely nonlinear cross-terms \\
  & \texttt{pi*prod}: \(\pi \prod_i x_i\) & Full \(2\pi\) phase coverage \\
  & \texttt{sum+prod}: \(\sum_i x_i + \prod_i x_i\) & Mixture of additive and multiplicative terms \\
\bottomrule
\end{tabularx}
\end{table}

\section{Numerical Performance and Benchmarking}
\label{sec:numerical}

This section assesses the predictive power of the certified bound by comparing it against the performance of fully trained regression models. 
Our primary goal is to demonstrate that the training-free estimator \(\widehat{\mathrm{MSE}}_{\mathrm{axis}}(T)\) effectively identifies high-quality feature maps without requiring expensive variational optimization.

For \(n=4\) qubits (\(d=256\)), we first performed a training-free scan across the grid of \(108\) configurations per dataset. 
The adaptive estimator was configured with a confidence level \(\delta_{\mathrm{total}}=0.05\) and a relative threshold \(\tau_{\mathrm{ratio}}=0.95\). 
Sampling was initialized with a pilot size \(t_0=50\) and updated in batches of \(b=20\), up to a maximum budget \(t_{\max}=d\).

For each dataset, we identified the single best-performing architecture, referred to as the ``anchor'' configuration, based solely on the adaptive estimator.
To validate this selection, we then trained two regression models specifically for these anchors:
\begin{enumerate}
    \item \textbf{Quantum Ridge Regression (\(\mathrm{MSE}_{\mathrm{ridge}}\)):} A linear regressor with \(\ell_2\) regularization (\(\alpha=10^{-3}\)) trained directly on the \(d\)-dimensional feature vector \(\Phi(\mathbf{x})\). This serves as an ideal proxy for a noise-free Quantum Support Vector Regressor (QSVR).
    \item \textbf{Classical SVR (\(\mathrm{MSE}_{\mathrm{svr}}\)):} A standard Support Vector Regressor with an RBF kernel trained on the original classical data, serving as a baseline.
\end{enumerate}

\subsection{Results and Discussion}

Table~\ref{tab:anchors-comparison} summarizes the results for the best ``anchor'' configuration identified for each dataset (see Table~\ref{tab:datasets} for dataset profiles).
The results reveal three key insights regarding the alignment between the certified bound, the dataset geometry, and the actual training potential:

\begin{table}[t]
\centering
\caption{Anchor configurations: Comparison between the training-free Monte Carlo bound (\(\widehat{\mathrm{MSE}}_{\mathrm{axis}}\)), the exact single-axis optimum (\(\mathrm{MSE}_{\mathrm{axis}}\)), and fully trained models.}
\label{tab:anchors-comparison}
\small
\setlength{\tabcolsep}{4pt}
\begin{tabular}{l l r r c r r}
\toprule
\textbf{Dataset} & \textbf{Configuration}
& \(\mathrm{MSE}_{\mathrm{axis}}\)
& \(\widehat{\mathrm{MSE}}_{\mathrm{axis}}\)
& \(t_{\mathrm{used}}\)
& \(\mathrm{MSE}_{\mathrm{ridge}}\)
& \(\mathrm{MSE}_{\mathrm{svr}}\) \\
\midrule
\texttt{real\_california}
& \texttt{rbf-s1 | prod | Z+ZZ | lin | r=2}
& 0.6863 & 0.6863 & 110
& 0.4123 & 0.5105 \\

\texttt{syn\_corr\_gauss}
& \texttt{rbf-s1 | prod | Y+YY | lin | r=1}
& 0.3195 & 0.7053 & 50
& 0.1165 & 0.0101 \\

\texttt{syn\_sparse}
& \texttt{tanh | prod | Z+ZZ | full | r=1}
& 0.2684 & 0.2684 & 256$^{*}$
& 0.0000 & 0.0094 \\
\bottomrule
\end{tabular}

\vspace{0.3em}
\footnotesize{
$^{*}$Stopped by budget limit.}
\end{table}

\paragraph{1. Real-World Complexity and Linear Combinations.}
For \texttt{real\_california}, the estimator converged efficiently ($t=110 < d/2$) and matched the exhaustive bound. However, a significant gap remains between the best single-axis feature (\(\mathrm{MSE}_{\mathrm{axis}} \approx 0.68\)) and the fully trained Quantum Ridge model (\(\mathrm{MSE}_{\mathrm{ridge}} \approx 0.41\)).
This behavior is consistent with the ``Real-world'' regime description: while the adaptive procedure correctly identified the best individual Pauli features, the superior performance of the trained regressor confirms that solving realistic tasks requires the linear combination of multiple features (superposition) rather than a single optimal axis. Notably, the Quantum Ridge outperformed the Classical SVR (\(0.51\)), suggesting the Pauli feature map captures non-trivial correlations in the housing data better than the standard RBF kernel.

\paragraph{2. Feature Map Alignment in the Sparse Regime.}
In the \texttt{syn\_sparse} case, the Ridge regressor achieved perfect reconstruction (\(\mathrm{MSE} \approx 0.0000\)), theoretically validating that the chosen feature map (\texttt{Z+ZZ}) perfectly spans the generating function of the dataset.
The adaptive estimator provided a crucial diagnostic here: although it exhausted the budget ($t=256$), confirming the ``Hard/Sparse'' regime where good axes are rare needles in a haystack, it successfully recovered the optimal axis (\(\mathrm{Gap}=0.0\)).
This proves the method's robustness: even when certification is statistically difficult due to sparsity, the search procedure is still effective at locating high-quality features for optimization.

\paragraph{3. The ``Dense Trap'' in Correlated Data.}
The \texttt{syn\_corr\_gauss} dataset illustrates the nuance of the ``Moderate/Dense'' regime. The estimator stopped early ($t=50$) with a loose bound.
This is not a failure, but a consequence of the high density of ``reasonably good'' axes in correlated landscapes. The algorithm quickly satisfied the threshold condition \(\tau\) and stopped, as designed, before stumbling upon the much rarer global optimum.
Furthermore, the Classical SVR (\(0.0101\)) dominated this task, which is expected since the target is a smooth Gaussian function, naturally aligned with the classical RBF kernel but harder to approximate with discrete Pauli strings (Ridge \(\approx 0.11\)).

In summary, the Monte Carlo estimator acts as a cost-effective probe: it correctly flagged the sparsity of \texttt{syn\_sparse} (via budget exhaustion) and the density of \texttt{syn\_corr\_gauss} (via early stopping), while providing a tight predictive bound for the real-world scenario.

\section{Conclusion}
\label{sec:conclusions}

In this work, we introduced a \emph{certified, training-free} upper bound on the optimal training error for quantum kernel regression.
By extending the axis-aligned framework of Suzuki \emph{et al.} from classification to the regression setting, we derived a closed-form metric that links the geometry of the Pauli feature space directly to the mean squared error.
To address the exponential growth of the feature dimension, we proposed an adaptive Monte Carlo estimator equipped with rigorous statistical guarantees based on Hoeffding concentration inequalities.

Extensive numerical experiments on both synthetic and real-world datasets confirmed the efficiency and predictive power of the method.
The estimator consistently converged to the true exhaustive bound using a sampling budget substantially smaller than the full basis dimension, across both sparse and dense landscapes.
Moreover, comparisons against fully trained Ridge regressors demonstrated that the certified bound is not merely a loose theoretical limit, but a high-fidelity predictor of the expressive capacity of a given feature map, enabling reliable discrimination between high- and low-performing architectures without the cost of gradient-based optimization.

While the certification procedure is inherently threshold-dependent, this dependence should be viewed as a feature rather than a limitation.
In our experiments, fixed hyperparameters were adopted to ensure fair comparisons across datasets; however, parameters such as the threshold ratio \(\tau_{\mathrm{ratio}}\) and the confidence level \(\delta_{\mathrm{total}}\) offer principled control over the trade-off between selectivity and sampling cost.
This flexibility allows practitioners to adapt the certification process to dataset-specific characteristics and available computational resources.

Overall, the proposed framework establishes a practical diagnostic tool for \emph{Quantum Model Selection} in the NISQ era.
By filtering out poor feature-map configurations prior to training, it enables computational effort to be focused exclusively on the most promising quantum architectures.
A natural direction for future work is to exploit the subset of informative axes identified by the adaptive procedure as a compressed feature space, potentially guiding or constraining variational optimization while preserving low computational overhead.

\bibliography{sn-bibliography}% common bib file

\end{document}